%% file: edge_coloring_streaming.tex
\documentclass[11pt,a4paper]{article}
\usepackage{amsmath,amssymb,amsfonts,latexsym,graphicx,amsthm}
\usepackage{color}
\usepackage{url,hyperref}
\usepackage{comment}
\usepackage[linesnumbered,boxed,ruled,noend]{algorithm2e}
\usepackage{xcolor} \usepackage{hyperref} \hypersetup{
    colorlinks = true, allcolors = black, citecolor = {blue!0!black}, linkcolor = {blue!0!black}
}
\usepackage[shortlabels]{enumitem}
\usepackage{cleveref}

\usepackage[margin=1in]{geometry}

\newtheorem{theorem}{Theorem}[section]
\newtheorem{lemma}{Lemma}[section]

\newtheorem{corollary}{Corollary}[section]

        {\medskip}

\newcommand{\eps}{\epsilon}

\newcommand{\ceil}[1]{\lceil #1 \rceil}
\newcommand{\floor}[1]{\lfloor #1 \rfloor}

\newcommand{\brac}[1]{\left(#1\right)}

\begin{document}

\title{Streaming Edge Coloring with Subquadratic Palette Size}
\author{Shiri Chechik\thanks{Tel Aviv University, 		\href{}{shiri.chechik@gmail.com}}\and
				Doron Mukhtar\thanks{Tel Aviv University, 		\href{}{doron.muk@gmail.com}}\and
				Tianyi Zhang \thanks{Tel Aviv University, \href{}{tianyiz21@tauex.tau.ac.il}}}
\date{}

\maketitle

\begin{abstract}

In this paper, we study the problem of computing an edge-coloring in the (one-pass) W-streaming model. In this setting, the edges of an $n$-node graph arrive in an arbitrary order to a machine with a relatively small space, and the goal is to design an algorithm that outputs, as a stream, a proper coloring of the edges using the fewest possible number of colors.

Behnezhad et al. [Behnezhad et al., 2019] devised the first non-trivial algorithm for this problem, which computes in $\tilde{O}(n)$ space a proper $O(\Delta^2)$-coloring w.h.p. (here $\Delta$ is the maximum degree of the graph). Subsequent papers improved upon this result, where latest of them [Ansari et al., 2022] showed that it is possible to deterministically compute an $O(\Delta^2/s)$-coloring in $O(ns)$ space.
However, none of the improvements succeeded in reducing the number of colors to $O(\Delta^{2-\eps})$ while keeping the same space bound of $\tilde{O}(n)$\footnote{$\tilde{O}(f)$ hides $\log^{O(1)} f$ factors.}. In particular, no progress was made on the question of whether computing an $O(\Delta)$-coloring is possible with roughly $O(n)$ space, which was stated in [Behnezhad et al., 2019] to be an interesting open problem.

In this paper we bypass the quadratic bound by presenting a new randomized $\tilde{O}(n)$-space algorithm that uses $\tilde{O}(\Delta^{1.5})$ colors. 
\end{abstract}

\thispagestyle{empty}
\clearpage
\setcounter{page}{1}

\input{intro}
\input{subq}

\section*{Acknowledgment}
This publication is part of a project that has received funding from the European Research Council (ERC) under the European Union’s Horizon 2020 research and innovation programme (grant agreement No 803118 UncertainENV).

\vspace{5mm}
\bibliographystyle{alpha}
\bibliography{ref}


\end{document}

%% file: intro.tex
\section{Introduction}


The last few decades have witnessed significant technological advancements, which have led to an exponential increase in the volume of data sets and network traffic that require efficient processing. However, many of the devices that we use to perform these tasks lack sufficient storage capacity to store even a small fraction of the input data, which typically arrives in an unordered stream. Consequently, we often find ourselves processing this data using partial information. This scenario is commonplace, especially when attempting to receive information from remote servers over the internet or fetch data from an external memory unit.

To gain a better understanding of computing capabilities in such scenarios, the \emph{data stream model} has been introduced. This model involves receiving an arbitrary stream of tokens as input, and the objective is to compute a function of this input by performing one or more passes over the stream while using minimal working memory. The model has been primarily applied to problems on graphs. In this context, the data stream consists of a sequence of updates that defines the edge-set of a graph with a known number of vertices, and the aim is to compute various properties of this graph.

Clearly, our aim in such problems is to design algorithms whose space complexity is asymptotically much less than the number of edges - preferably linear in the number of vertices. However, in some cases, the size of the output may be as large as the size of the input, forcing us to use a very large space to store it. To get around this, we use a known variant of the streaming model, called the W-steaming model \cite{wstream}, which allows us to stream parts of the output along the computation. 
However, it forces the algorithm to commit to parts of the output without even seeing the entire input.

Many important problems have been studied in the W-streaming model. In this paper we focus on the problem of properly coloring the edges of a given graph with a small number of colors. This problem is considered to be one of the most fundamental and well-studied problems in graph theory, with many applications in scheduling and communications.

Behnezhad et al. \cite{behnezhad2019streaming} were the first to consider the problem of edge-coloring in the W-streaming model. They distinguished between two variants: random order streams -- in which the edges arrive according to a random permutation that was chosen uniformly at random before the start of the algorithm, and adversarial order streams -- in which the edges arrive in an arbitrary order. For the first variant, they provided a simple one-pass algorithm that uses $O(\Delta)$ colors and $\tilde{O}(n)$ space (where as usual $\Delta$ and $n$ respectively denote the maximum degree of the input graph and the number of its vertices). For the second one, they provided a different one-pass algorithm that in $\tilde{O}(n)$ space computes w.h.p. a proper $O(\Delta^2)$-coloring. Charikar and Liu \cite{charikar2021improved} improved the above results by devising a one-pass $\tilde{O}(n)$-space algorithm that uses $\Delta + o(\Delta)$ colors for random order streams, and a one-pass randomized algorithm that uses $\tilde{O}(ns)$ space and $(1+o(1))\Delta^2/s$ colors for adversarial order streams. More recently, Ansari et al. \cite{ansari2022simple} provided two simple deterministic algorithms that in one pass and $O(ns)$ space compute a proper coloring that uses no more than $(1+o(1))\Delta^2/s$ colors for adversarial order streams, improving the result of \cite{charikar2021improved}.

Interestingly, for the more challenging model of adversarial order streams, all of the above results require $\tilde{\Omega}(\Delta^2)$ colors when the available space is $\tilde{O}(n)$. This raises the question of whether one can get asymptotically below this number of colors, while retaining the same space bound, or if $\tilde{\Omega}(\Delta^2)$ is an inherent limitation. In this paper we resolve this question, and show that it is possible to reduce the number of colors without compromising on increasing the given space.

\subsection{Our result}
We break the quadratic barrier for the first time by providing a new randomized algorithm that in one pass and $\tilde{O}(n)$-space (in expectation) computes a proper edge-coloring which uses, in expectation, no more than $\tilde{O}(\Delta^{1.5})$ colors (for adversarial order streams).

\begin{theorem}\label{subq}
	For an undirected (multi-)graph $G = (V, E)$ on $n$ vertices and maximum degree $\Delta$, there is a single-pass randomized streaming algorithm using $\tilde{O}(n)$ space for edge coloring that uses $\tilde{O}(\Delta^{1.5})$ colors with high probability.
\end{theorem}

\subsection{Related work}

The edge-coloring problem has been studied in several models of computation (see, e.g. \cite{MISRA, HJKD, YQWJ, behnezhad2019streaming}). A closely related model to the W-Streaming is known as the Online model. In the online edge-coloring problem, we also assume that the edges of the input graph arrive as a data stream (which can be either random or arbitrary). There is no space limitation, but there is a requirement to instantly compute and output a color for each newly received edge, as opposed to the W-Streaming model in which we can delay the announcement of some of the edges by buffering them. See \cite{srinivasan2023online,kulkarni2022online,WaSA,SAWA,cohen2019tight} for some latest results on this problem.

\paragraph*{Independent work} In a concurrent work \cite{behnezhad2023streaming}, Behnezhad and Saneian have obtained the same result as \Cref{subq}; besides, they have a general trade-off of $\tilde{O}(ns)$ space and $\tilde{O}(\Delta^{1.5}/s)$ colors, for any parameter $s$. In another concurrent work \cite{ghosh2023low}, Ghosh and Stoeckl have achieved a trade-off of $\tilde{O}(ns)$ space and $\tilde{O}(\Delta^{2}/s^2)$ colors, for any parameter $s$.

\subsection{Technical Overview}

Ansari et al. \cite{ansari2022simple} used, in their second algorithm, a very simple buffering approach to color the graph's edges. Let $G = (V, E)$ be the input graph, $n$ be its number of vertices and $\Delta$ be its maximum degree. Assuming that we have a working space of at least $n$ words, we divide the edge-stream $E$ into continuous intervals $E_1,...,E_k$ of $n$ edges each, buffer each one of them in order, and color each graph $G_i = (V,E_i)$ separately by using a different set of $O(\Delta)$ colors. This gives us a proper $O(\Delta^2)$-coloring as there can be at most $\Delta$ such intervals.

To reduce the number of colors in such an approach, we have to avoid somehow the allocation of a new $O(\Delta)$-color palette for each interval. Our main observation here is that we may not be able to edge-color the graph of a given interval $G_i$ with much less than $\Omega(\Delta)$ colors (as it may contain vertices of high degree), but we can use fewer colors to color all the edges in $G_i$ that are between vertices of a sufficiently low vertex degrees. We will still have to use palettes of $O(\Delta)$ colors to color the rest of the edges, but these edges are now adjacent to high degree vertices. As the number of these vertices in each of the intervals cannot be too large, we can afford to store in the machine more information about the colorings that were used to color these edges across several intervals (which we call a phase). This way, it will hopefully be sufficient to allocate a few palettes of $O(\Delta)$-colors per phase, instead of allocating a new one for each interval.

Let us develop this idea further. We divide the intervals $E_1,...,E_k$ into contiguous phases $P_1,...,P_{k'}$ of $D$ intervals each. In each phase $i \in \{1,...,k'\}$, we are going to process the intervals one by one, and color their edges. To do that, we distinguish, in each interval $E_j \in P_i$,  between the ``low'' edges - those whose endpoints are of low degree (at most some parameter $t = \Theta(\sqrt{\Delta})$) in $G_j$, the ``high'' edges - those whose endpoints are of high degree (greater than $t$ in $G_j$) and the ``mixed'' - all the rest. As we discussed above, in each interval $E_j \in P_i$, the number of nodes whose degree in $G_j$ is greater than $t$, cannot be that large (less than $2n/t$). Thus, as long as $t \ge D$, we can store for the current phase, $O(1)$ words for each node and each interval such that this node is of degree greater than $t$ in this interval.

With that in mind, consider the following procedure for coloring the high edges. We allocate at the start of the phase, $c$ palettes $A_1,...,A_c$ of $O(\Delta)$ colors each. In each interval, we choose uniformly at random a set $A$ from the collection, and color all the high edges that are incident only to vertices that weren't previously incident to high edges that were colored using $A$. (that is, we store for each node that was incident to a high edge, the IDs of all the palettes that were used to color its incident edges.) Note that each vertex may have a high degree in no more than $\Delta/t$ intervals. Therefore, the probability that a high edge will not be colored in a certain interval is at most $2\Delta/tc$, which for $c = C(\Delta/t)$ is $2/C$. This means that in expectation only a fraction of the edges remains uncolored. Let us ignore these edges for a moment, and discuss how we color the rest of the edges.

Coloring the low edges of each interval is simple as we don't need big palettes for them. So we just allocate a new set of $O(t)$ colors for each interval, and properly color the low edges of that interval with it.

We move now for the task of coloring the mixed edges. Note that this case is more complicated than the previous ones, as we need to use palettes of $K = O(\Delta)$-colors, but cannot store too much information for the nodes with low degree. As before, we start by assigning at the start of the phase a set of $c = C(\Delta/t)$ palettes $B_1,...,B_c$ of $O(\Delta)$ colors each. In each interval, we choose uniformly at random a set $B$ from the collection, and the idea now is that the vertex with the low degree of each mixed edge will choose a color from $B$ to color this edge.

We illustrate the problems that could arise (when trying to color these edges) and the way we handle them, by focusing on two edges $\{u,v\}$ and $\{u,w\}$ that belong to two different intervals. Note that since these edges share a vertex, they cannot be colored with the same color. In the first case, $u$ has a high degree in both of the intervals. A possible conflict can occur if we happen to chose the same set $B$ in both of them. To solve this, we do the same thing that we did for the high edges (i.e. we store at the high degree vertices the IDs of the palettes with which previous incident edges were colored, and avoid coloring an edge if a conflict was detected). This will guarantee that in expectation only a fraction of them remains uncolored. In the second case, $u$ has a low degree in both of the intervals. This means that both $w$ and $v$ have a high degree in their intervals, but they may be different, and so we cannot know the IDs of the previous palettes. To solve this, we maintain for $u$ a counter $c_u$, which counts the total number of edges that connected $u$ to high degree vertices when it had low degree. When $u$ wants to choose a color from $B$ it simply takes the $c_u$ color in it, this will guarantee that no conflict can occur in such cases. In the third and last case, $u$ has a high degree in one of the intervals and a low in the other. For resolving possible conflicts in this case, we make use of random offsets. At the start of the phase, we choose for each vertex $v$ a uniformly random integer $r_v$ between $0$ and $K-1$, and use it to permute the colors in the set (that is, $v$ considers the first color in $B$ to be the color in the position $r_v$, and so on). Before we color an edge $\{u,v\}$ we check whether the offsets $r_v$ and $r_u$ are far enough. If this is not the case we skip it (this happens with small probability).

To summarize, in each phase we use $O(c\Delta + tD)$ colors. As we have $O(\Delta/D)$ phases, we get that the total number of colors is $O(\Delta^3/(tD) + t\Delta)$. The edges that were left uncolored, we treat as a new virtual stream, and apply recursively the algorithm on it. As in expectation the number of uncolored edges reduces each time by a constant factor, we get that it only increases the space bound and the number of colors by a factor of $O(\log\Delta)$ .

%% file: subq.tex
\section{Edge coloring with subquadratic palette size}
\subsection{Algorithm description}
\paragraph*{Preparation} Our algorithm assumes prior knowledge of the maximum degree $\Delta$. Without loss of generality, assume $\sqrt{\Delta}$ is an integer power of $2$. If the data stream contains at most $O(n)$ edges, then the entire stream fits in the memory, so we can color the graph with $O(\Delta)$ colors straightforwardly (more precisely, $\ceil{\frac{3\Delta}{2}}$ colors for multi-graphs, and $\Delta+1$ colors for simple graphs \cite{shannon1949theorem,vizing1965critical}).

The algorithm divides the data stream into intervals, each interval consisting of $n$ edges from $E$ (the last interval may have less than $n$ edges). In general, there are at most $|E| / n = O(\Delta)$ intervals in the data stream. Let $E_l\subseteq E$ be the set of edges in the $l$-th interval. Denote by $G_l = (V, E_l)$ the subgraph collected from this interval. We further partition all intervals into \emph{phases}, each phase consisting of $\Delta^{1/2}$ consecutive intervals (the last phase may have less). Different phases will use different sets of edge colors and thus are independent. For the rest of this section, we will describe our algorithm for an individual phase.

Fix any degree threshold $d$ be any integral power of $2$. If $d < \sqrt{\Delta}$, we will create $O(\sqrt{\Delta})$ new colors for each interval which should be enough to color all edges $(u, v)$ such that $$\max\{\deg_{G_l}(u), \deg_{G_l}(v)\} < 2d$$
Otherwise if $d\geq \sqrt{\Delta}$, then for each interval $E_l$, we will describe an algorithm that assigns colors to all edges $(u, v)$ such that $$\max\{\deg_{G_l}(u), \deg_{G_l}(v)\}\in [d, 2d)$$
Ranging over all choices of $d$, we would be able to color the whole graph $G_l$, by blowing up the total number of colors by a factor of $O(\log\Delta)$. 

A vertex $v\in V$ is called \emph{high-degree} in $G_l$, if $\deg_{G_l}(v) \in [d, 2d)$, and if $\deg_{G_l}(v) < d$, then it is called \emph{low-degree}; we will not consider any vertices whose degree in $G_l$ is at least $2d$.

Let $\kappa\geq 32$ be a constant which is an integer power of $2$. At the beginning of a single phase, the algorithm prepares three sequences of disjoint color palettes: $A_1, A_2, \cdots, A_{\kappa\Delta / d}$, and $B_1, B_2, \cdots, B_{\kappa\Delta / d}$, and $C_1, C_2, \cdots, C_{\kappa\Delta / d}$. All palettes $A_i, B_i, C_i$ have size $K =2\kappa d$.

For each palette $X$ ($X\in \{A_i, B_i, C_i\}$), its colors will be indexed by $0, 1, 2\cdots, K-1$. Each $A_i$ will only be used to color edges between high-degree vertices in some intervals, and each $B_i\cup C_i$ will only be used to color edges between high-degree vertices and low-degree vertices in some intervals. For each vertex $v\in V$, take a uniformly random integer offset $r_v\in [0, K-1]$. Additionally, for each palette index $i$, maintain a counter $c_v[i]\in \{\bot\}\cup [0, K-1]$ which will bound the number of edges incident on $v$ with colors from palette $C_i$; when $c_v[i]= \bot$, it means this counter $c_v[i]$ has not been initialized yet, so it does not take any space in the storage.

During this phase, each vertex $v$ holds a set of indices $I_v\subseteq [\kappa\Delta/d]$, such that $i\in I_v$ if edges incident on $v$ have been colored with colors from palette $A_i$ in any previous intervals. Finally, the algorithm maintains a \emph{virtual stream} of \emph{leftover} edges, which are edges discarded temporarily but will be processed again recursively.

\paragraph*{Processing intervals} Let $l_0$ be the index of the first interval of this phase. For each interval indexed by $l$, draw a random index $\sigma_l\in [\kappa\Delta /d]$ uniformly at random. For any $i$, keep a counter $p_l[i] = |\{\sigma_k = i\mid l_0\leq k< l\}|$. Note that all these counters $\{p_l[i]\}_{1\leq i\leq \kappa\Delta/d}$ can be stored using $O(\kappa\Delta / d) = O(n)$ space; for the $l$-th interval, we are only using counters $p_l[*]$, and previous counters will be discarded.

Partition the graph $G_l = (V, E_l)$ into two disjoint subgraphs: $G_l = H_l^1\cup H_l^2$ defined below; edges in these subgraphs will be colored separately.
\begin{enumerate}[(1),leftmargin=*]
	\item $H_l^1$ consists of all edges between high-degree vertices in $G_l$.	
	
	To color edges in $H_l^1$, let $U$ be the set of all high-degree vertices such that $\sigma_l\notin I_v$. Then, use the palette $A_{\sigma_l}$ to color the subgraph $H_l^1[U]$ since $|A_{\sigma_l}| = 2\kappa d > 4d-1$.
	
	After that, for each high-degree vertex $v$, add $\sigma_l$ to $I_v$ by updating $I_v\leftarrow I_v\cup \{\sigma_l\}$. Finally, for each high-degree vertex $v\notin U$, insert all its incident edges in $G_l$ to the virtual stream as leftover edges; we emphasize the point that this also includes its edges connecting to low-degree vertices.

	\item $H_l^2$ consists of all edges connecting high-degree vertices to low-degree vertices.
	
	To color edges in $H_l^2$, go over all low-degree vertices $u$ in the alphabetical order. For each low-degree vertex $u$, if $c_u[\sigma_l]=\bot$, then check if $\deg_{G_l}(u) > \sqrt{\Delta}$; if so, then initialize a counter $c_u[\sigma_l] \leftarrow 0$.
	
	Enumerate its incident edges in $H_l^2$  in the alphabetical order. For the $b$-th edge $(u, v)\in E_l$ incident on $u$ (starting with $b=0$), consider several cases below.
	\begin{enumerate}[(a),leftmargin=*]
		\item If $(u, v)$ was already inserted to the virtual stream in Step (1) due to the reason that $v\notin U$, then skip it and proceed to the next edge incident on $u$.
		
		\item Define $\delta = r_v - r_u \mod K$, so $d\in [0, K-1]$. If $\delta < 2d$ or $\delta>K-2d$, then insert edge $(u, v)$ to the virtual stream as a leftover edge.
		
		\item Otherwise, assume $2d\leq \delta\leq K-2d$. There are several sub-cases.
		\begin{itemize}[leftmargin=*]
			\item If $c_u[\sigma_l] = 2d$, then simply insert the edge $(u, v)$ to the virtual stream as a leftover edge.
			
			\item If $c_u[\sigma_l]\neq \bot$ and $c_u[\sigma_l] < 2d$, then check if the $(r_u + c_u[\sigma_l])$-th color from palette $C_{\sigma_l}$ has already been used in $H_l^2$ on any other edge incident on high-degree vertex $v$. If so, then insert the edge $(u, v)$ to the virtual stream as a leftover edge; otherwise, assign the $(r_u + c_u[\sigma_l])$-th color from palette $C_{\sigma_l}$ to edge $(u, v)$.
			
			\item Otherwise if $c_u[\sigma_l] = \bot$, check if $p_l[\sigma_l] \geq 2d / \sqrt{\Delta}$. If so, insert edge $(u, v)$ to the virtual stream as a leftover edge.
			
			In the case that $c_u[\sigma_l] = \bot$ and $p_l[\sigma_l] < 2d / \sqrt{\Delta}$, check if the $\brac{r_u + b + \sqrt{\Delta}\cdot p_l[\sigma_l]}$-th color from palette $B_{\sigma_l}$ has already been used in $H_l^2$ on any other edge incident on high-degree vertex $v$. If so, then insert the edge $(u, v)$ to the virtual stream as a leftover edge; otherwise, assign the $\brac{r_u + b + \sqrt{\Delta}\cdot p_l[\sigma_l]}$-th color from palette $B_{\sigma_l}$ to edge $(u, v)$.
		\end{itemize}
	\end{enumerate}
	In all cases (a)(b)(c), if $c_u[\sigma_l]\neq \bot$, increment the counter by one $c_u[\sigma_l]\leftarrow c_u[\sigma_l]+1$ afterwards (even if $(u, v)$ is inserted to the virtual stream).
\end{enumerate}

\paragraph*{Coloring leftover edges} To color all leftover edges in the virtual stream, we apply the above algorithm recursively on this virtual stream as its input data stream, with fresh colors and independent randomness. Some of the notations used in the algorithm description are summarized in \Cref{table:notation-summary}.
\begin{table}
	\centering
	\begin{tabular}{| c | c |}
		\hline 
		\textbf{notation} & \textbf{definition} \\ \hline
		$E_l$	&	the set of $O(n)$ edges in the $l$-th interval in from the data stream\\\hline
		
		$G_l$	&	$G_l = (V, E_l)$ is the subgraph collected in the $l$-th interval\\\hline
		
		phase	&	a phase consists of $\Delta^{1/2}$ consecutive intervals\\\hline
		
		$\kappa$	&	a constant at least $32$ which is also an integer power of $2$\\\hline
		
		$A_{i}, B_{i}, C_{i}, 1\leq i\leq \kappa\Delta / d$	&	three sequences of sets of palettes of size $K = 2\kappa d$\\\hline
		
		$c_v[i]\in \{\bot\}\cup [0, K-1]$	&	a counter bounding using colors from palette $C_i$\\\hline
		
		$r_v$	&	$r_v\in [0, K-1]$ is a random shift\\\hline
		
		$I_v$	&	$i\in I_v$ if $v$ already has edges with colors from $A_i$\\\hline
		
		$\sigma_l$	&	a uniformly random index from $[1, \kappa\Delta/d]$ for the $l$-th interval\\\hline
		
		$p_l[i]$	&	a counter defined by $p_l[i] = |\{\sigma_k = i\mid l_0\leq k< l\}|$\\\hline
		
		high-degree vertex	&	high-degree if $\deg_{G_l}\in [d, 2d)$\\\hline
		
		low-degree vertex	&	high-degree if $\deg_{G_l} <d$\\\hline
		
		$H_l^1$	&	all edges between high-degree vertices in $G_l$\\\hline
		
		$H_l^2$	&	all edges connecting high-degree vertices and low-degree vertices\\\hline
	\end{tabular}
	\caption{Some of the notations used in the algorithm.}
	\label{table:notation-summary}
\end{table}


\subsection{Proof of correctness}
Let us first state some basic properties of the algorithm.

\begin{lemma}\label{ind}
	The value of any counter $c_u[i]$ depends on the data stream and randomness of indices $\{\sigma_l\}$, but is independent of the random shifts $\{r_v\mid v\in V\}$.
\end{lemma}
\begin{proof}
	According to the description of the algorithm, each counter $c_u[i]$ is initialized to $0$ in the first interval of edges $E_k$ such that $\deg_{G_k}(u) > \sqrt{\Delta}$ and $\sigma_k = i$; this event is independent of random shifts $\{r_v\mid v\in V\}$. Later on, for any $l\geq k$ such that $\sigma_l = i$, $c_u[i]$ will increase by $|\deg_{H^2_l}(u)|$ if $u$ is a low-degree vertex in $G_l$. Therefore, $c_u[i]$ is always independent of the random shifts $\{r_v\mid v\in V\}$.
\end{proof}

\begin{lemma}\label{C-range}
	At the beginning of the $l$-th interval, for any low-degree vertex $u$, if $c_u[\sigma_l]\neq \bot$, then all colors from any palette $C_{\sigma_l}$ indexed by $r_u + c_u[\sigma_l], r_v + c_u[\sigma_l]+1, \cdots, r_u+  2d-1\mod K$ haven't been used for any edge incident on $u$.
\end{lemma}
\begin{proof}
	This is because the counter $c_u[\sigma_l]$ increases by one each time we use a color from $C_{\sigma_l}$ on edges incident on $u$, so the colors indexed by larger integers have not been used yet.
\end{proof}

\begin{lemma}\label{full}
	For each vertex $u$, at any moment during the current phase, there are at most $\frac{\Delta}{2d}$ different indices $i$ such that $$\sum_{l, \sigma_l = i}\deg_{G_l}(u) \geq 2d$$
\end{lemma}
\begin{proof}
	Since $\deg_{G}(u)\leq \Delta$, the number of indices $i$ such that $\sum_{l, \sigma_l = i}\deg_{G_l}(u)\geq 2d$ is bounded by $\frac{\Delta}{2d}$.
\end{proof}

Next, we show that our algorithm always produces a valid edge coloring.
\begin{lemma}
	The algorithm always outputs a valid edge coloring of $G$.
\end{lemma}
\begin{proof}
	Consider any two adjacent edges $(u, v), (u, w)\in E$, and we will show that $(u, v), (u, w)$ must have different colors in the output stream. Note that $v$ does not need to be different from $w$ since $G$ might not be a simple graph. If one of them appears in the virtual stream, then because the recursion for leftover edges uses fresh colors for each interval, we can make sure that $(u, v), (u, w)$ are assigned different colors. For the rest we assume that $(u, v), (u, w)$ are not leftover edges.
	
	First, consider the case where both $(u, v), (u, w)$ appear in the same interval from the input stream. According to the algorithm, if they are colored in different Step (1) and (2) respectively, then they are picking colors from separate palettes; if they are colored in the same Step (1) or (2), then the algorithm must have colored them in a compatible way.
	
	Secondly, consider the case where $(u, v), (u, w)$ appear in different intervals from the input stream. Assume $(u, v)$ is in interval $E_l$ and $(u, w)$ is in interval $E_h, h\neq l$. Consider the following cases.
	\begin{itemize}[leftmargin=*]
		\item One of $(u, v), (u, w)$ is colored in Step (1) and the other in Step (2).
		
		In this case, they are using different palettes, so their colors are always different.
		
		\item Both $(u, v), (u, w)$ are colored in Step (1) but in different intervals.
		
		In this case, since $(u, w)$ is not a leftover edge, by the algorithm we know that $\sigma_l\neq \sigma_h$. Therefore, $(u, v), (u, w)$ use colors from different palettes $A_{\sigma_l}, A_{\sigma_h}$.
		
		\item Both $(u, v), (u, w)$ are colored in Step (2) but in different intervals.
		
		
		If $u$ is high-degree in both intervals $E_l, E_h$, then since neither of $(u, v), (u, w)$ is  a leftover edge, it must be $\sigma_l \neq \sigma_h$. Hence, $(u, v), (u, w)$ are taking colors from different palettes.
		
		Next, assume $u$ is low-degree in both intervals, plus that $\sigma_l = \sigma_h$. If $(u, v), (u, w)$ were selecting colors from $B_{\sigma_l}, C_{\sigma_h}$, then their colors must be different. So, we only need to consider the case where both $(u, v), (u, w)$ were selecting colors from $B_{\sigma_l}$, or from $C_{\sigma_l}$. Let us analyze these two cases separately.
		
		\begin{itemize}[leftmargin=*]
			\item $(u, v), (u, w)$ were selecting colors from $B_{\sigma_l}$.
			Then, by the algorithm description, $(u, v)$ could only take colors from $B_{\sigma_l}$ taking indices from the range
			$$[r_u + \sqrt{\Delta}\cdot p_l[\sigma_l], r_u+\sqrt{\Delta}\cdot \brac{p_l[\sigma_l]+1}-1] \mod K$$
			and $(u, w)$ could only take colors from $B_{\sigma_l}$ taking indices from the range
			$$[r_u + \sqrt{\Delta}\cdot p_h[\sigma_h], r_u+\sqrt{\Delta}\cdot \brac{p_h[\sigma_h]+1}-1] \mod K$$ 
			
			As $l\neq h$ and $\sigma_l = \sigma_h$, we know that $p_l[\sigma_l]\neq p_h[\sigma_h]$, and therefore the two ranges $$[r_u + \sqrt{\Delta}\cdot p_l[\sigma_l], r_u+\sqrt{\Delta}\cdot \brac{p_l[\sigma_l]+1}-1]\mod K$$ and $$[r_u + \sqrt{\Delta}\cdot p_h[\sigma_h], r_u+\sqrt{\Delta}\cdot \brac{p_h[\sigma_h]+1}-1]\mod K$$ are disjoint, and thus the colors of the two edges must be different.
			
			\item $(u, v), (u, w)$ were selecting colors from $C_{\sigma_l}$.
			
			Suppose the color indices of $(u, v), (u, w)$ were $r_u+ c_u^{(1)}[\sigma_l]$ and $r_u + c_u^{(2)}[\sigma_l]$, where $c_u^{(1)}[\sigma_l]\neq c_u^{(2)}[\sigma_l]$ are the counter values of $c_u[\sigma_l]$ by the time when the color was selected. Hence, $(u, v), (u, w)$ are picking colors with different indices.
		\end{itemize}
	
		Finally, assume $u$ is high-degree in $G_l$ and low-degree in $G_h$. Then, when coloring the edge $(u, v)$, by the condition on Step (2)(b), we know that $\delta = r_u - r_v\mod K$ belongs to $[2d, K-2d]$. By \Cref{C-range}, as the color index of $(u, v)$ always belongs to $\{r_v, r_v+1, \cdots, r_v+2d-1 \}$ and the color index of $(u, w)$ always belongs to $\{r_u, r_u+1, \cdots, r_u+2d-1 \}$, the colors of $(u, v), (u, w)$ should be different as the two index sets are disjoint.
	\end{itemize}
\end{proof}

Next, we verify that the total space of each recursion level is at most $O(n)$.
\begin{lemma}\label{mem}
	The space usage of each recursion level of the algorithm is $O(n)$.
\end{lemma}
\begin{proof}
	It is clear that in each phase, the total space of all offsets is at most $O(n)$. For the number of counters in a single phase, every time a new counter $c_u[\sigma_l]$ has been initialized, we must have $\deg_{G_l}(u) > \sqrt{\Delta}$ for some interval index $l$ within this phase. Since the total sum of vertex degrees within this phase is bounded by $O(n\sqrt{\Delta})$, we know that the number of counters is bounded by $O(n)$.
	
	As for the total space sets $I_v$'s, in each interval, any set $I_v$ adds one more element if $v$ is high-degree. As the total number of edges collected from a single phase is $O(n\sqrt{\Delta})$, The total size of $\sum_{v\in V}|I_v|$ would be at most $O(n\sqrt{\Delta} / d) = O(n)$.
\end{proof}

To bound the total number different colors in $G$, we need to analyze the total number of leftover edges of each recursion level.
\begin{lemma}\label{leftover}
	Suppose there are $m$ edges from the input stream (across all phases) in total. Then, the expected number of leftover edges to enter the virtual stream is at most $7m/\kappa$, over the randomness of shifts $\{r_v\}_{v\in V}$ and indices $\{\sigma_l\}$.
\end{lemma}
\begin{proof}
	Consider any single phase, and let $(u, v)\in E$ be an arbitrary edge that appears in interval $E_l$. There are several cases where $(u, v)$ could possibly be a leftover edge that gets inserted to the virtual stream.
	\begin{itemize}[leftmargin=*]
		\item Edge $(u, v)$ becomes leftover on Step (1).
		
		The condition for $(u, v)$ being cast as a leftover edge is that $\sigma_l\in I_u\cup I_v$; that is, palette $A_{\sigma_l}$ was already used for edges incident on $u$ or $v$ in a previous interval. Since any vertex can be high-degree in at most $\Delta /d$ intervals, the number of non-zero entries of $I_u, I_v$ is at most $2\Delta/d$. As $\sigma_l\in [\kappa\Delta/d]$ is uniformly random, the probability that $\sigma_l\in I_u\cup I_v$ is at most $2/\kappa$, over the randomness of $\sigma_l$.
		
		\item Edge $(u, v)$ becomes leftover on Step (2) because $c_u[\sigma_l] = 2d$.
		
		By \Cref{full}, there are at most $\frac{\Delta}{2d}$ different indices $i$ such that $\sum_{k=l_0, \sigma_k = i}^{l-1}\deg_{G_k}(v)\geq 2d$; let $I$ be the set of these indices. Then, since $\sigma_l$ is picked uniformly at random, the probability that $\sigma_l\in I$ is at most $\frac{1}{2\kappa}$.
		
		As $c_u[i]\leq \sum_{k=l_0, \sigma_k = i}^{l-1}\deg_{G_k}(v)$ for any $i$, the probability that $c_u[\sigma_l] = 2d$ is at most $\frac{1}{2\kappa}$. Therefore, $(u, v)$ becomes leftover on this step with probability at most $\frac{1}{2\kappa}$.
		
		\item Edge $(u, v)$ becomes leftover on Step (2)(b).
		
		In this case, we have $\delta < 2d$ or $\delta>K-2d$ for $\delta = r_v - r_u\mod K$. Since $r_v, r_u$ were chosen uniformly at random from $[0, K-1]$, the probability of this event is at most $4d / K < 2 / \kappa$, over the randomness of $r_v$.
		
		\item Edge $(u, v)$ becomes leftover on Step (2)(c) due to $p_l[\sigma_l] \geq 2d / \sqrt{\Delta}$.

		Note that since each $\sigma_k$ was picked uniformly at random, the expectation of $p_l[\sigma_l]$ which counts $|\{\sigma_k = \sigma_l \mid l_0\leq k<l\}|$ is bounded by $\frac{d}{\kappa\sqrt{\Delta}}$, over the randomness of $\sigma_{l_0}, \sigma_{l_0+1}, \cdots, \sigma_{l-1}$. By Markov's inequality, the probability that $p_l[\sigma_l] \geq 2d / \sqrt{\Delta}$ is at most $\frac{1}{2\kappa}$.
		
		\item Edge $(u, v)$ becomes leftover on Step (2)(c) and $p_l[\sigma_l] < 2d / \sqrt{\Delta}$.
		
		In this case, assume $u$ is low-degree and $v$ is high-degree in interval $E_l$. We need to consider two more sub-cases below.
		
		\begin{itemize}[leftmargin=*]
			\item $c_u[\sigma_l] \neq \bot$.
			
			In this case, edge $(u, v)$ was attempting to use a color from $C_{\sigma_l}$. Right before $(u, v)$ was enumerated, let $S$ be the set of low-degree neighbors $w$ of $v$ whose alphabetical orders are before $u$, plus that $c_w[\sigma_l]\neq \bot$.
			
			For each vertex $w\in S$, suppose there are $k_w$ parallel edges between $v, w$, and let $b_w$ be the value of the counter $c_w[\sigma_l]$ when the first edge $(u, w)$ was enumerated from the perspective of $w$. Define an index set $$I = \bigcup_{w\in S}\{r_w+b_w, r_w + b_w+1, \cdots, r_w + b_w + k_w-1 \}$$
			
			Note that $|I|< 2d$ since $\sum_{w\in S}k_w\leq \deg_{G_l}(v)<2d$. As the algorithm enumerates vertices and edges on Step (2)(c) in the alphabetical order, we know that $I$ is the set of all possible color indices of edges $(v, w), w\in S$ in palette $C_{\sigma_l}$.
			
			Now, on the one hand, by \Cref{ind}, all values of $b_w, k_w$ are independent of the random shifts $\{r_z\mid z\in V\}$. Therefore, as $u\notin S$, we know that $r_u$ is independent of $I$. On the other hand, for $(u, v)$ to be a leftover edge, $r_u + c_u[\sigma_l]$ must belong to $I$. As the value $c_u[\sigma_l]$ is also independent of $r_u$, the probability that $r_u + c_u[\sigma_l]\in I$ is at most $|I| / K < 1/\kappa$, over the random choice of $r_u$.
			
			\item $c_u[\sigma_l] = \bot$.
			
			In this case, edge $(u, v)$ was attempting to use a color from $B_{\sigma_l}$. Similar to the previous sub-case, right before $(u, v)$ was enumerated, let $T$ be the set of low-degree neighbors $w$ of $v$ whose alphabetical orders are before $u$, plus that $c_w[\sigma_l] = \bot$.
			
			For each vertex $w\in T$, suppose there are $k_w$ parallel edges between $u, w$, and $(u, w)$ (the first copy) is the $b_w$-th edge incident on $w$. Define an index set 
			$$J = \bigcup_{w\in T}\{r_w+b_w+\sqrt{\Delta}\cdot p_l[\sigma_l], r_w + b_w+\sqrt{\Delta}\cdot p_l[\sigma_l]+1, \cdots, r_w + b_w+\sqrt{\Delta}\cdot p_l[\sigma_l] + k_w-1 \}$$
			
			Note that $|J|< 2d$ since $\sum_{w\in T}k_w\leq \deg_{G_l}(v)<2d$. As the algorithm enumerates vertices and edges on Step (2)(c) in the alphabetical order, we know that $J$ is the set of all possible color indices of edges $(v, w), w\in T$ in palette $B_{\sigma_l}$.
			
			Suppose $(u, v)$ is the $b$-th edge of $u$ in $G_l$, so $r_u+b+\sqrt{\Delta}\cdot p_l[\sigma_l]$ is the color index that $(u, v)$ attempted to use in $B_{\sigma_l}$. Since $u\notin T$, and that $r_u$ is independent of $p_l[\sigma_l]$ and all values in $\{b_w\mid w\in T\}$, we know that the probability that $r_u + b+\sqrt{\Delta}\cdot p_l[\sigma_l]\in J$ is at most $|J| / K < 1/\kappa$, over the random choice of $r_u$.
		\end{itemize}
	\end{itemize}

	Taking a summation of all the cases, the probability that $(u, v)$ becomes a leftover edge is at most $7/\kappa$. Therefore, expected number of leftover edges is bounded by $7m/\kappa$.
\end{proof}

\begin{lemma}\label{exp-depth}
	The expected recursion depth of the main algorithm is $O(\log\Delta)$.
\end{lemma}
\begin{proof}
	Consider any recursion where the input stream contains $m$ edges. By Markov's inequality and \Cref{leftover}, at each recursion level, with probability at least $1/2$, the number of leftover edges is at most $14m/\kappa < m/2$. Since the original graph $G$ contains $O(n\Delta)$ edges, after $O(\log\Delta)$ recursion levels in expectation, the input has at most $O(n)$ edges, so the algorithm would not recurse further.
\end{proof}

As a corollary, we can bound the total number of different colors and total memory, which concludes the proof of \Cref{subq}.
\begin{corollary}
	The number of colors used by our algorithm is $\tilde{O}\brac{\Delta^{1.5}}$, and the total memory is bounded by $O(n\log\Delta)$; both bounds hold in expectation.
\end{corollary}
\begin{proof}
	Consider any recursion level. According to the algorithm, in each phase and each choice of $d\geq \sqrt{\Delta}$ which is an integer power of $2$, the total number of colors in the palettes $\{A_i\}\cup \{B_i\}\cup \{C_i\}$ is $O(\Delta)$. Also, each interval creates at most $O(\sqrt{\Delta})$ new colors when $d<\sqrt{\Delta}$. Since there are $O(\sqrt{\Delta})$ phases, the number of colors used for processing intervals is $\tilde{O}(\Delta^{1.5})$, summing over all choices of $d$. As the expected recursion depth is $O(\log\Delta)$, together with \Cref{mem} we can finish the proof.
\end{proof}

\subsection{Adaptation to an unknown $\Delta$}
So far we have assumed that the value of the maximum degree $\Delta$ of $G$ is known to the algorithm as prior knowledge. We can adapt our algorithm to an unknown $\Delta$ by losing a constant factor in the total number of colors in the following way. Basically, we will maintain the value $\Delta_t$ which is the maximum degree of the subgraph containing the first $t$ edges in the data stream. Whenever $\Delta_t\in (2^{k-1}, 2^{k}]$, we will apply \Cref{subq} with $\Delta = 2^k$ to color all the edges. If $k$ increases at some point, we will restart a new instance of \Cref{subq} with a new choice of $\Delta = 2^k$ and continue to color the edges with a separate palette; to clarify, when a new instance of \Cref{subq} is restarted, we continue with the current pass of the data stream, not starting over with a new pass. In the end, the total number of colors will be $\tilde{O}(\sum_{k=1}^{\ceil{\log\Delta}}2^{1.5k}) = \tilde{O}(\Delta^{1.5})$ in expectation.

\subsection{From expectation to high probability}
So far our bounds on memory and the number of colors only hold in expectation, not with high probability. These bounds can actually hold with high probability rather than in expectation. In fact, as shown in \Cref{exp-depth}, the expected recursion depth can be bounded by $O(\log\Delta)$. We can also show that the depth is $O(\log n)$ with high probability, and therefore the total number of colors is at most $O(\Delta^{1.5}\cdot\log n)$. To get rid of the dependency on $\log n$, we can increase the size of intervals by a factor of $\log n$; that is, each interval now contains $O(n\log n)$ edges. This would decreases the total number of phases by a factor of $\Omega(\log n)$ and thus decreases the number of colors, which cancels out the extra $\log n$ factor in the color bound.